\documentclass{llncs}  
\usepackage{amsmath}
\usepackage{amssymb}
\usepackage[dvips]{graphicx}
\usepackage{pstricks}
\usepackage{verbatim}
\usepackage{psfrag}

\usepackage{booktabs}

\spnewtheorem{mylemma}[theorem]{Lemma}{\bf}{\it}
\spnewtheorem{myprop}[theorem]{Proposition}{\bf}{\it}

\newcommand{\rank}[0]{\textrm{rank}}
\newcommand{\score}[0]{\textrm{score}}
\newcommand{\wt}[1]{\widetilde{#1}}

\begin{document}

\mainmatter
\title{Multivariate Complexity Analysis of Swap Bribery}

\author{ Britta Dorn\inst{1} \and Ildik\'o Schlotter\inst{2}\thanks{Supported 
by the Hungarian National Research Fund (OTKA 67651).}}

\institute{Fakult\"at f\"ur Mathematik und Wirtschaftswissenschaften, Universit\"at Ulm\\
Helmholtzstr.~18, D-89081 Ulm, Germany\\
\email{britta.dorn@uni-ulm.de}
\and 
Budapest University of Technology and Economics\\
H-1521 Budapest, Hungary\\
\email{ildi@cs.bme.hu}}

\maketitle

\begin{abstract}
We consider the computational complexity of a problem modeling {\it bribery} in the context of voting systems.
In the scenario of {\sc Swap Bribery}, each voter assigns a certain price for swapping the positions of two 
consecutive candidates in his preference ranking. 
The question is whether it is possible, without exceeding a given budget, 
to bribe the voters in a way that the preferred candidate wins in the election. 

We initiate a parameterized and multivariate complexity analysis of {\sc Swap Bribery}, 
focusing on the case of $k$-approval. We investigate how different cost functions affect
the computational complexity of the problem. We identify a special case of $k$-approval 
for which the problem can be solved in polynomial time, whereas we prove NP-hardness 
for a slightly more general scenario. 
We obtain fixed-parameter tractability 
as well as W[1]-hardness results for certain natural parameters. 
\end{abstract}


\section{Introduction}\label{sec:introduction}

In the context of voting systems, the question of how to manipulate the votes 
in some way in order to make a preferred candidate win the election is a very 
interesting question. One possibility is {\it bribery}, which can be described 
as spending money on changing the voters' preferences over the candidates in 
such a way that a preferred candidate wins, while respecting a given budget. 
There are various situations that fit into this scenario: 
The act of bribing the voters in order to make them change their preferences, 
or paying money in order to get into the position of being able to change the 
submitted votes, but also the setting of systematically spending money in an 
election campaign in order to convince the voters to change their opinion 
on the ranking of candidates.

The study of bribery in the context of voting systems was initiated by Faliszewski, 
Hemaspaandra, and Hemaspaandra in 2006~\cite{FHH06}. Since then, various models have 
been analyzed. In the original version, each voter may have a different but fixed 
price which is independent of the changes made to the bribed vote. 
The scenario of nonuniform bribery introduced by Faliszewski~\cite{Fal08} 
and the case of microbribery studied by Faliszewski, Hemaspaandra, Hemaspaandra, 
and Rothe in~\cite{FHHR09} allow for prices that depend on the amount of change 
the voter is asked for by the briber.

In addition, the {\sc Swap Bribery} problem as introduced by Elkind, Faliszewski, 
and Slinko~\cite{EFS09} takes into consideration the ranking aspect of the votes: 
In this model, each voter may assign different prices for swapping two consecutive 
candidates in his preference ordering. This approach is natural, since it captures 
the notion of small changes and comprises the preferences of the voters.
Elkind et al.~\cite{EFS09} prove complexity results for this problem for several 
election systems such as Borda, Copeland, Maximin, and approval voting. 
In particular, they provide a detailed case study for $k$-approval. 
In this voting system, every voter can specify a group of~$k$ preferred candidates 
which are assigned one point each, whereas the remaining candidates obtain no points. 
The candidates which obtain the highest sum of points over all votes are the winners 
of the election. Two prominent special cases of~$k$-approval are plurality, 
(where $k =1$, i.e., every voter can vote for exactly one candidate) and 
veto (where $k = m-1$ for $m$ candidates, i.e., every voter assigns one point 
to all but one disliked candidate).
Table~\ref{tab:t-approval} shows a summary of research considering {\sc Swap Bribery} 
for $k$-approval, including both previously known and newly achieved results.

\begin{table}[t]
\begin{center}
    \begin{tabular}{@{ }llll@{ }}
   \toprule
      & Result & Reference \\
\midrule
$k=1$ or $k=m-1$  & P  & \cite{EFS09} \\
$k=2$  & NP-complete   & \cite{BD10,EFS09} \\
$k \geq 3$ constant,   & NP-complete  & \cite{EFS09} \\
	\quad costs in $\{0,1,2\}$ & & \\
$k \geq 2$ constant,  & {\bf NP-complete}  & \cite{BD10}, Prop.~\ref{costs_01}\\
	\quad costs in $\{0,1\}$ and $\beta=0$  & & \\
$k \geq 2$ constant,  & {\bf NP-complete}, {\bf W[1]-hard ($\beta$)} &  Thm.~\ref{thm_2approval}\\
	\quad costs in $\{1,1+\varepsilon \}$, $\varepsilon > 0$ & & \\
$k$ constant, $m$ or $n$ constant  & P   & \cite{EFS09} 
\\
$k$ part of the input, $n=1$ & NP-complete & \cite{EFS09}\\
$2 \leq k \leq m-2$ part of the input, & {\bf NP-complete}  & \cite{BHN09}, Prop.~\ref{costs_01}\\
	\quad costs in $\{0,1\}$ and $\beta=0$, $n$ constant & & \\ 
$k$ part of the input, $n=1$ & {\bf W[1]-hard ($k$)} & Thm.~\ref{thm_k_hard}\\
$k$ part of the input, all costs $=1$  & {\bf P}   & Thm.~\ref{uniform}   \\
$k$ part of the input & {\bf FPT ($m$)} &  Thm.~\ref{thm_candidates}\\ 
$k$ part of the input & {\bf FPT ($n$)} & Thm.~\ref{thm_color}\\
$k$ part of the input & {\bf FPT ($\beta,n$) by kernelization}  &  Thm.~\ref{thm_kernel}\\ 
    \bottomrule
    \end{tabular}
\end{center}
\caption{
Overview of known and new results for {\sc Swap Bribery} for~$k$-approval. 
The results obtained in this paper are printed in bold. 
Here, $m$ and $n$ denote the number of candidates and votes, respectively, and~$\beta$ is the budget. 
For the parameterized complexity results, the parameters are indicated in parentheses. 
}\label{tab:t-approval}
\end{table}

This paper contributes to the further investigation of the case study of $k$-approval 
that was initiated in~\cite{EFS09}, this time from a parameterized point of view. 
%
The main goal of this approach is to find {\it fixed-parameter tractable} (FPT) algorithms confining 
the combinatorial explosion which is inherent in NP-hard problems to certain 
problem-specific parameters, or to prove that their existence is implausible. 
This line of research has been pioneered by Downey and Fellows~\cite{DF99}, 
see also~\cite{FG06,Nie06} for two more recent monographs, and naturally expands into 
the field of multivariate algorithmics, where the influence of ``combined'' 
parameters is studied, see the recent survey by Niedermeier~\cite{Nie10}. 
These approaches
seem to be appealing in the context of voting systems, where NP-hardness 
is a desired property for various problems, like {\sc Manipulation} (where certain voters, 
the manipulators, know the preferences of the remaining voters and try to adjust their own 
preferences in such a way that a preferred candidate wins), {\sc Lobbying} (here, a 
lobby affects certain voters on their decision for several issues in an election), {\sc Control} 
(where the chair of the election tries to make a certain candidate win (or lose) by deleting 
or adding either candidates or votes), or, as in our case, {\sc Swap Bribery}. 
However, NP-hardness does not necessarily constitute a guarantee against such dishonest behavior. 
As Conitzer et al.~\cite{CSL07} point out for the {\sc Manipulation} problem, 
an NP-hardness result in these settings would lose relevance if an efficient 
fixed-parameter algorithm with respect to an appropriate parameter was found. 
Parameterized complexity can hence provide a more robust notion of hardness. 
The investigation of problems from voting theory under this aspect has started, 
see for example~\cite{Bet10,BHN09,BU09,CFRS07,LFZL09}.

We examine how the computational complexity of {\sc Swap Bribery} for $k$-approval
depends on certain restrictions on the cost function.
We show NP-hardness as well as fixed-parameter intractability of {\sc Swap Bribery} 
for a very restricted version of the problem with a fixed value of~$k$ if the 
parameter is the budget, whereas we identify a natural special case 
which can be solved in polynomial time.
By contrast, we obtain fixed-parameter tractability with respect to the parameter 
`number of candidates' for $k$-approval and a large class of other voting systems. 
We also investigate the parameter `number of votes', and consider the situation where~$k$ 
is part of the input, for which {\sc Swap Bribery} is known to be NP-complete already 
for only one vote. We strengthen this result by proving W[1]-hardness with respect 
to the parameter~$k$, whereas we obtain fixed-parameter tractability with respect to~$n$ 
for the case where~$k$ is a constant by using the technique of color-coding.
We also provide a polynomial kernel where we consider certain combined parameters.
%

The paper is organized as follows. 
After introducing notation in Section~\ref{sec:preliminaries}, we investigate the complexity 
of {\sc Swap Bribery} depending on the cost function in Section~\ref{sec:costfunction}, 
where we show the connection to the {\sc Possible Winner} problem, 
identify a polynomial-time solvable case of $k$-approval and a hardness result. 
In Section~\ref{sec:candidates}, we consider the parameter `number of candidates' 
and obtain an FPT result for {\sc Swap Bribery} for a large class of voting systems. 
Section~\ref{sec:votes} investigates the influence of the parameter `number of votes', 
providing both W[1]-hardness and fixed-parameter tractability results and considering 
combinations of different parameters.
We conclude with a discussion of open problems and further directions 
that might be interesting for future investigations.


\section{Preliminaries}\label{sec:preliminaries}
{\bf Elections.} An $\mathcal{E}$-{\it election} is a pair $E= (C,V)$, 
where $C= \{c_1, \dots, c_m\}$ denotes the set of {\it candidates}, 
$V=\{v_1, \dots, v_n\}$ is the set of {\it votes} or {\it voters}, 
and $\mathcal{E}$ is the {\it election system} which is a function mapping $(C,V)$ 
to a set~$W\subseteq C$ called the {\it winners} of the election. 
We will express our results for the {\it winner case} where several winners are possible, but our results 
can be adapted to the {\it unique winner case} where~$W$ consists of a single candidate only. 

In our context, each vote is a strict linear order over the set~$C$, and we denote 
by $\textrm{rank}(c,v)$ the position of candidate~$c \in C$ in a vote~$v\in V$. 
By contrast, the concept of partial votes, mentioned only occasionally in this paper, 
can be used to describe partial orders over the candidates.

For an overview of different election systems, we refer to~\cite{BF02}. 
We will mainly focus on election systems that are characterized by a given {\it scoring rule}, 
expressed as a vector~$(s_1, s_2, \dots, s_m)$. Given such a scoring rule, the {\it score} 
of a candidate~$c$ in a vote~$v$, denoted by $\textrm{score}(c,v)$, is $s_{\textrm{rank}(c,v)}$. 
The score of a candidate~$c$ in a set of votes~$V$ is $\textrm{score}(c,V) = \sum_{v \in V}\textrm{score}(c,v)$, 
and the winners of the election are the candidates that receive the highest score in the given votes. 

The election system we are particularly interested in is $k$-approval, 
which is defined by the scoring vector $(1, \dots, 1, 0, \dots, 0)$, starting with~$k$~ones. 
In the case of $k=1$, this is the {\it plurality} rule, whereas $(m-1)$-approval is also known as {\it veto}. 
Given a vote~$v$, we will say that a candidate~$c$ with $1 \leq \textrm{rank}(c,v) \leq k$ takes 
a {\it one-position} in~$v$, whereas a candidate~$c'$ with $k+1 \leq \textrm{rank}(c',v) \leq m$ 
takes a {\it zero-position} in~$v$.

\vspace{4pt}
\noindent
{\bf Swap Bribery, Possible Winner, Manipulation.}
Given $V$ and $C$, a \emph{swap} in some vote~$v \in V$ is a triple $(v,c_1, c_2)$ where 
$\{c_1, c_2\} \subseteq C, c_1 \neq c_2$. 
Given a vote $v$, we say that a swap~$\gamma=(v,c_1, c_2)$ is \emph{admissible in $v$}, 
if $\textrm{rank}(c_1,v) = \textrm{rank}(c_2,v)-1$. 
Applying this swap means exchanging the positions of $c_1$ and $c_2$ in the vote $v$, 
we denote by $v^{\gamma}$ the vote obtained this way. 
Given a vote $v$, a set $\Gamma$ of swaps is \emph{admissible in $v$}, if 
the swaps in $\Gamma$ can be applied in~$v$ in a sequential manner, one after the other, in some order. 
Note that the obtained vote, denoted by $v^{\Gamma}$, is independent from the order in which the swaps of $\Gamma$ are applied.
We also extend this notation for applying swaps in several votes, in the straightforward way. 

In a \textsc{Swap Bribery} instance, we are given $V$, $C$, and $\mathcal{E}$ forming an election, 
a preferred candidate $p \in C$, a cost function $c:C\times C \times V \rightarrow \mathbb{N}$ 
mapping for every vote each possible swap to a non-negative integer, 
and a budget $\beta \in \mathbb{N}$. The task is to determine a set of admissible 
swaps~$\Gamma$ whose total cost is at most~$\beta$, such that $p$ is a winner in 
the $\mathcal{E}$-election $(C, V^{\Gamma})$. Such a set of swaps is called a 
\emph{solution} of the \textsc{Swap Bribery} instance. 
The underlying decision problem is the following. 

\vspace{-4pt}
\begin{quote}
\textsc{Swap Bribery}\\
\textbf{Given:} An $\mathcal{E}$-election $E=(C,V)$, a preferred candidate~$p\in C$, 
a cost function~$c$ mapping each possible swap to a non-negative integer, and a budget~$\beta \in \mathbb{N}$.\\
\textbf{Question:} Is there a set of swaps~$\Gamma$ whose total cost is at 
most~$\beta$ such that~$p$ is a winner in the $\mathcal{E}$-election~$(C, V^{\Gamma})$? 
\end{quote}
\vspace{-4pt}

We will also show the connection between {\sc Swap Bribery} and the {\sc Possible Winner} problem. 
In this setting, we have an election where some of the votes may be {\it partial} orders over~$C$ 
instead of complete linear ones. The question is whether it is possible to extend the partial votes 
to complete linear orders in such a way that a preferred candidate wins the election. 
For a more formal definition, we refer to the article by Konczak and Lang~\cite{KL05} who introduced this problem. 
The corresponding decision problem is defined as follows. 

\vspace{-4pt}
\begin{quote}
\textsc{Possible Winner}\\
\textbf{Given:} A set of candidates~$C$, a set of partial votes $V'= (v'_1, \dots, v'_n)$ over~$C$, 
an election system~$\mathcal{E}$, and a preferred candidate~$p \in C$.\\
\textbf{Question:} Is there an extension $V=(v_1, \dots, v_n)$ of $V'$ such that 
each~$v_i$ extends~$v'_i$ to a complete linear order, and $p$ is a winner in the $\mathcal{E}$-election $(C, V)$?
\end{quote}
\vspace{-4pt}

A special case of {\sc Possible Winner} is {\sc Manipulation} (see e.g.~\cite{CSL07,HH07,EFS09}). 
Here, the given set of partial orders consists of two subsets; 
one subset contains complete preference orders and the other one completely unspecified votes. 
\\

\noindent
{\bf Parameterized complexity, Multivariate complexity.} 

Parameterized complexity is a two-dimensional framework for studying the computational complexity of
problems~\cite{DF99,FG06,Nie06}. One dimension is the size of the input $I$ 
(as in classical complexity theory) and the other dimension is the parameter $k$ (usually a positive integer). 
A problem is called {\it fixed-parameter tractable} (FPT) with respect to a 
parameter~$k$ if it can be solved in~$f(k) \cdot |I|^{O(1)}$ time, 
where~$f$ is an arbitrary computable function~\cite{DF99,FG06,Nie06}. 
Multivariate complexity is the natural sequel of the parameterized approach 
when expanding to multidimensional parameter spaces, see~\cite{Nie10}.
For example, if we consider a parameterization where the parameter is a pair $k=(k_1,k_2)$, 
then we refer to this as a \emph{combined parameterization}, 
or we simply say that both $k_1$ and $k_2$ are parameters.
In such a case, the desired FPT algorithm should run in time~$f(k_1,k_2) \cdot |I|^{O(1)}$ for some~$f$.

The first level of (presumable) parameterized intractability is captured by the complexity class~W[1]. 
A {\it parameterized reduction} reduces a problem instance $(I, k)$ in $f(k)\cdot|I|^{O(1)}$ time 
(for some computable function~$f$) to an instance~$(I',k')$   
such that~$(I,k)$ is a yes-instance if and only if~$(I',k')$ is a yes-instance, 
and~$k'$ only depends on~$k$ but not on~$|I|$. 
To prove W[1]-hardness of a given parameterized problem $Q$,
one needs to present a parameterized reduction from some already known W[1]-hard problem to $Q$.

We will use the following W[1]-hard 
problems~\cite{downey-fellows-clique-hardness,fellows-hermelin-rosamond-vialette-multicolored-hardness} 
for the hardness reductions in this work: 
\vspace{-4pt}
\begin{quote}\textsc{Clique}\\
\noindent
\textbf{Given:} An undirected graph~$G=(V, E)$ and $k \in \mathbb{N}$.\\
\textbf{Question:} Is there a complete subgraph (clique) of~$G$ of size~$k$?
\end{quote}
\vspace{-4pt}

\vspace{-4pt}
\begin{quote}\textsc{Multicolored Clique}\\
\noindent
\textbf{Given:} An undirected graph~$G=(V_1 \cup V_2 \cup \dots \cup V_k, E)$ with  
$V_i \cap V_j = \emptyset$  for $1 \leq i < j \leq k$ where the vertices of $V_i$ 
induce an independent set for $1 \leq i \leq k$.\\
\textbf{Question:} Is there a complete subgraph (clique) of~$G$ of size~$k$?
\end{quote}
\vspace{-4pt}

We will also make use of a {\it kernelization} algorithm in this work, 
which is a standard technique for obtaining fixed-parameter results, see \cite{Bod09,GN07,Nie06}. 
The idea is to transform the input instance~$(I,k)$ in a polynomial-time preprocessing step via 
{\it data reduction rules} into a ``reduced'' instance~$(I', k')$ such that two conditions hold: 
First, $(I,k)$ is a yes-instance if and only if~$(I',k')$ is a yes-instance, and second, 
the size of the reduced instance depends on the parameter only, 
i.e. $|I'|+|k'| \leq g(p)$ for some arbitrary computable function~$g$. 
The reduced instance~$(I',k')$ is then referred to as the {\it problem kernel}. 
If in addition~$g$ is a polynomial function, we say that the problem admits a {\it polynomial kernel}. 
The existence of a problem kernel is equivalent to fixed-parameter tractability of the 
corresponding problem with respect to the particular parameter~\cite{Nie06}.



\section{Complexity depending on the cost function}\label{sec:costfunction}

In this section, we focus our attention on {\sc Swap Bribery} for $k$-approval. 
We start with the case where all costs are equal to~$1$, for which we obtain polynomial-time solvability. 
Below we provide an algorithm which for every possible $s$ checks
if there is a solution in which the preferred candidate wins with score $s$.
This can be carried out by solving a minimum cost maximum flow problem.

\begin{theorem}
\label{uniform}
\textsc{Swap Bribery} for $k$-approval is polynomial-time solvable, if all costs are $1$.
\end{theorem}

\begin{proof}
Let $V$ be the set of votes and $C$ be the set of candidates. 
The score of any candidate is an integer between $0$ and $|V|$.
Our algorithm finds out for each possible $s^*$ with $1 \leq s^* \leq |V|$ whether 
there is a solution in which the preferred candidate $p$ wins with score $s^*$.

\begin{figure}[t]
\begin{center}
    \psfrag{A}[lb][lb]{$A$}
    \psfrag{A'}[lb][lb]{$A'$}
    \psfrag{B}[lb][lb]{$B$}
    \psfrag{s}[lb][lb]{$s$}
    \psfrag{s*}[lb][lb]{$\scriptstyle s^*$}
    \psfrag{Vk-s*}[lb][lb]{$\scriptstyle |V|k-s^*$}
    \psfrag{x}[lb][lb]{$x$}
    \psfrag{t}[lb][lb]{$t$}
    \psfrag{bp}[lb][lb]{$b_p$}
    \psfrag{bc1}[lb][lb]{$b_{c_1}$}
    \psfrag{bc2}[lb][lb]{$b_{c_2}$}
    \psfrag{bc3}[lb][lb]{$b_{c_3}$}
    \psfrag{bc4}[lb][lb]{$b_{c_4}$}
	\psfrag{avc1}[lb][lb]{$a_{v,c_1}$}	
	\psfrag{a'vp}[lb][lb]{$a'_{v,p}$}	
	\psfrag{awc2}[lb][lb]{$a_{u,c_2}$}	
	\psfrag{a'wc4}[lb][lb]{$a'_{u,c_4}$}	
\includegraphics[scale=0.5]{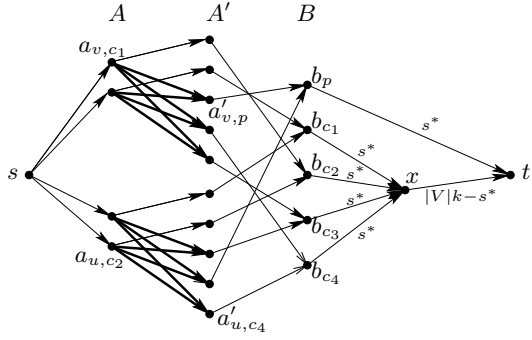}
\end{center}
\caption{The network for a small instance with $k=2$, having 5 candidates and 2 voters. 
The votes of the illustrated instance are $v: c_1 \succ c_2 \succ p \succ c_4 \succ c_3$ 
and $u: c_1 \succ c_2 \succ c_3 \succ p \succ c_4$.
%
Each unlabeled edge has capacity $1$; otherwise labels correspond to capacities. 
Bold edges correspond to bribing voters 	 
in order to move a point from one candidate to another candidate.
The costs of such edges depend on the given votes, e.g. 
$w(a_{v,c_1} a'_{v,p})=1$ and $w(a_{u,c_2} a'_{u,c_4})=3$. All other edges have cost $0$.
}
\label{fig_flow}
\end{figure}

Given a value $s^*$, we answer the above question by solving a corresponding minimum cost maximum flow problem. 
We will define a network $N=(G,s,t,g,w)$ on a directed graph $G=(D,E)$ with a source vertex $s$ and a target vertex $t$, 
where $g$ denotes the capacity function and $w$ the cost function defined on $E$.
See Figure~\ref{fig_flow} for an illustration of the network.
First, we introduce the vertex sets 

\vspace{2pt}
\begin{tabular}{l}
$A = \{ a_{v,c} \mid v \in V, c \in C, \textrm{rank}(c,v) \leq k \}$, \\
$A' = \{ a'_{v,c} \mid v \in V, c \in C \}$, and \\
$B = \{ b_c \mid c \in C \}$, 
\end{tabular}
\vspace{2pt}

\noindent
and we set $D=\{s,t,x \} \cup A \cup A' \cup B$.
We define the arcs $E$ as the union of the sets

\vspace{2pt}
\begin{tabular}{l}
$E_S = \{ sa \mid a \in A\}$, \\
$E_A = \{ a_{v,c} a'_{v,c} \mid \textrm{rank}(c,v) \leq k \}$, \\
$E_{A'} = \{ a_{v,c} a'_{v,c'} \mid \textrm{rank}(c,v) \leq k , \textrm{rank}(c',v) > k \}$,  \\
$E_B = \{ a'_{v,c} b_c \mid v \in V, c \in C \}$,  \\
$E_X = \{b_c x \mid c \in C, c \neq p \}$, 
\end{tabular}
\vspace{2pt}

\noindent
plus the arcs $b_{p} t$ and $xt$. 
We set the cost function $w$ to be $0$ on each arc except for the arcs of $E_{A'}$, and we set 
$w(a_{v,c} a'_{v,c'}) = \textrm{rank}(c',v) - \textrm{rank}(c,v)$.
We let the capacity $g$ be $1$ on the arcs of $E_S \cup E_A \cup E_{A'} \cup E_B$, 
we set it to be $s^*$ on the arcs of $E_X \cup \{ b_{p} t \}$, 
and we set $g(xt)=|V|k-s^*$. 

The soundness of the algorithm and hence the theorem itself follows from the following observation:
there is a flow of value $|V|k$ on $N$ having total cost at most~$\beta$ 
if and only if there exists a set~$\Gamma$ of swaps with total cost at most $\beta$ such that
$\mathrm{score}(p,V^{\Gamma}) = s^*$ and $\mathrm{score}(c,V^{\Gamma}) \leq s^*$ for any $c \in C, c \neq p$.

First, suppose that such a flow~$f$ exists. 
Since all capacities and costs are integrals, we know that $f$ is integral as well. 
For each vote $v \in V$, we define a set of swaps on~$v$ as follows.
We define two sets $X^{\rightarrow}(v)$ and $X^{\leftarrow}(v)$ in a way that
if $f(a_{v,c} a'_{v,c'})=1$ holds for some $c$ and $c'$ with $c \neq c'$, then we put $c$ into $X^{\rightarrow}(v)$
and we put $c'$ into $X^{\leftarrow}(v)$. 
Clearly, $|X^{\rightarrow}(v)| = |X^{\leftarrow}(v)|$, by the given capacities.
Observe that moving the candidates in $X^{\rightarrow}(v)$ to the positions $k+1, k+2, \dots, k+h$ and also the 
candidates in $X^{\leftarrow}(v)$ to the positions $k, k-1, \dots, k-h+1$ for $h = |X^{\rightarrow}(v)|$ 
has total cost $\sum_{c' \in X^{\leftarrow}(v)} \mathrm{rank}(c',v) - \sum_{c \in X^{\rightarrow}(v)} \mathrm{rank}(c,v)$.
Thus, by letting $\Gamma(v)$ contain these swaps for some $v$, we know that the cost of the bribery 
$\Gamma = \{ \Gamma(v) \mid v \in V \}$ is exactly the cost of the flow $f$ which is not more than $\beta$. 
Observe that as a result of these swaps, each candidate $c$ other than $p$ will 
receive at most $s^*$ scores in $V^{\Gamma}$ because of the capacity $g(b_c x) \leq s^*$.
On the other hand, by $g(xt) = |V|k-s^*$ we get $f(b_{p}t)=s^*$, 
which yields that $p$ will receive exactly $s^*$ scores in $V^{\Gamma}$.
Thus, $\Gamma$ has the properties claimed.

For the converse direction, let $\Gamma$ be a set of swaps with total cost at most $\beta$ such that
$\mathrm{score}(p,V^{\Gamma}) = s^*$ and $\mathrm{score}(c,V^{\Gamma}) \leq s^*$ for any $c \in C, c \neq p$.
For some $v \in V$, let $X^{\rightarrow}(v)$ denote those candidates $c$ for which 
$\mathrm{score}(c,v)>\mathrm{score}(c,v^{\Gamma})$, and let 
$X^{\leftarrow}(v)$ denote those candidates $c'$ for which $\mathrm{score}(c',v)<\mathrm{score}(c',v^{\Gamma})$. 
It is easy to see that the swaps applied in $v$ by $\Gamma$ have total cost at least 
$\sum_{c \in X^{\leftarrow}(v)} \mathrm{rank}(c',v) - \sum_{c \in X^{\rightarrow}(v)} \mathrm{rank}(c,v)$. 
Therefore, a flow can be easily constructed having cost at most $\beta$ 
in the following way: for each $v$ and~$c$ where $\mathrm{score}(c,v)=\mathrm{score}(c, v^{\Gamma})=1$ we let 
$f(a_{v,c} a'_{v,c})=1$, and for each  $v$ and $c$ where $c$ is the $i$-th candidate in $X^{\rightarrow}(v)$
we set  $f(a_{v,c} a'_{v,c'})=1$ for the $i$-th candidate $C'$ of $X^{\leftarrow}(v)$ according to some fixed ordering.
It is not hard to verify that this indeed determines a flow for $N$, with value $|V|k$ and cost at most $\beta$.  
\qed
\end{proof}

Note that Theorem~\ref{uniform} implies a polynomial-time approximation algorithm 
for \textsc{Swap Bribery} for $k$-approval with approximation ratio $\delta$, 
if all costs lie within the range $[1, \delta]$ for some $\delta \geq 1$.

Proposition~\ref{costs_01} shows the connection between {\sc Swap Bribery} and {\sc Possible Winner}. 
This result is an easy consequence of a reduction given by Elkind et al. \cite{EFS09}. 

\begin{myprop}
\label{costs_01}
The special case of \textsc{Swap Bribery} where the costs are in $\{0,\delta\}$ for some~$\delta>0$ and the budget is zero is equivalent to the \textsc{Possible Winner} problem. 
\end{myprop}

\begin{proof}
It has already been proved by Elkind et al. \cite{EFS09}
that \textsc{Possible Winner} reduces to \textsc{Swap Bribery} if the possible costs include $0$ and $1$, and 
the budget is zero. 
Clearly, the result also holds if we assume that the costs include $0$ and $\delta$ for some $\delta>0$. 

For the other direction, it is easy to see that a  \textsc{Swap Bribery} instance with costs in $\{0,\delta\}$, $\delta>0$ 
and budget zero is equivalent to the \textsc{Possible Winner} instance with the same candidates where each vote $v$ is replaced 
by the transitive closure of the relation $\succ_v$ for which $a \succ_v b$ holds 
if and only if $a$ precedes $b$ in the vote $v$ and the cost of swapping $a$ with $b$ in $v$ is non-zero.
\qed
\end{proof}

As a corollary, \textsc{Swap Bribery} with costs in $\{0,\delta\}$, $\delta>0$ and 
budget zero is NP-complete for almost all election systems based on scoring vectors~\cite{BD10},
and also for the voting rules Copeland~\cite{LX08} and Maximin~\cite{LX08}. 
For many voting systems such as $k$-approval, Borda, and Bucklin, 
it is NP-complete even for a fixed number of votes \cite{BHN09}.
A further consequence of Proposition~\ref{costs_01}, 
contrasting the polynomial-time approximation algorithm implied by Theorem~\ref{uniform}, 
is the fact that approximating {\sc Swap Bribery} with an arbitrary factor 
in a setting where zero costs are allowed
is NP-hard for all voting rules where the {\sc Possible Winner} problem is NP-hard.
This has been observed by Elkind and Faliszewski~\cite{EF-approx} as well.

We now turn our attention to the simplest case among those where the cost function is not a constant, i.e. 
where only two different positive costs are possible. 
Theorem~\ref{thm_2approval} shows that the corresponding problem is hard already for $2$-approval.

\begin{theorem}
\label{thm_2approval} Suppose that~$\varepsilon > 0$. \\
(1) \textsc{Swap Bribery} for 2-approval with costs in~$\{1,1+\varepsilon\}$ is NP-complete. \\
(2) \textsc{Swap Bribery} for 2-approval with costs in~$\{1,1+\varepsilon\}$ is W[1]-hard, 
if the parameter is the budget~$\beta$, or equivalently, the maximum number of swaps allowed.
\end{theorem}

\begin{proof}
We present a reduction from the \textsc{Multicolored Clique} problem. 
Let~$\mathcal{G}=(V,E)$ with the~$k$-partition~$V=V_1 \cup V_2 \cup \dots \cup  V_k$ be 
the given instance of \textsc{Multicolored Clique}. 
Here and later, we write~$[k]$ for~$\{1,2, \dots, k\}$.
For each~$i \in [k]$,~$x \in V_i$, and~$j \in [k] \setminus \{i\}$ we let 
$E_x^{j} = \{ xy \mid y \in V_j, xy \in E\}$.
We construct an instance~$I_{\mathcal{G}}$ of \textsc{Swap Bribery} as follows.

The set~$\mathcal{C}$ of candidates will be the union of the sets 
$A, B, C, \wt{C}, F, H, \wt{H}, M, \wt{M}, D, G, T$, and~$\{p,r \}$, Table~\ref{table_candidates}
shows the exact definition of these sets.
Our preferred candidate is~$p$. The sets~$D$, $G$, and $T$
will contain \emph{dummies}, \emph{guards}, and \emph{transporters}, respectively. 
Our budget will be~$\beta=k^3+10k^2$.
Regarding the indices~$i$ and~$j$, we suppose~$i, j \in [k]$ if not stated otherwise.

\begin{table}[t]
\begin{center}
\begin{tabular}{l@{\hspace{10pt}}l}
candidate set & cardinality  \\
\noalign{\hrule}
\\[-8pt]
$A = \{a^{i,j} \mid i, j \in [k] \}$ 			& $|A|=k^2$ 	\\
$B = \{b_v^j \mid j \in [k], v \in V\}$			& $|B|=k|V|$	\\
$C = \{c_v^j \mid j \in [k], v \in V\}$ 		& $|C|=k|V|$	\\
$\wt{C} = \{\wt{c}_v^j \mid j \in [k], v \in V\}$ & $|\wt{C}|=k|V|$	\\
$F = \{f_v^j \mid j \in [k], v \in V\}$			& $|F|=k|V|$	\\
$H = \{h_v^j \mid j \in [k], v \in V\}$			& $|V|=k|V|$	\\
$\wt{H} = \{h_v^j \mid j \in [k], v \in \bigcup_{i<j} V_i \}$	& $|\wt{H}|=\sum_{j=1}^{k} (k-j)|V_j| < k|V|$	\\
$M = \{m^{i,j} \mid 1 \leq i \leq j \leq k \}$	& $|M|=\binom{k}{2}+k$	\\
$\wt{M} = \{\wt{m}^{i,j} \mid 1 \leq i < j \leq k \}$	& $|\wt{M}|=\binom{k}{2}$	\\
$D=\{ d_1 , d_2, \dots \}$						& $|D| \leq |W| = O(k^3|V|^2)$	\\
$G =\{ g_1, g_2, \dots, g_{\beta+2} \}$			& $|G|=\beta+2=k^3+10k^2$	\\
$T=\{ t_1, t_2, \dots \}$ 						& $|T| =O(k^2|V|)$	
\\[1mm]
\noalign{\hrule}
\end{tabular}
\end{center}
\caption{The candidate sets constructed in the proof of Theorem~\ref{thm_2approval}.}
\label{table_candidates}
\end{table}

The set of votes will be~$W= W_G \cup W_I \cup W_S \cup W_C$. 
Votes in~$W_G$ will define guards (explained later),
votes in~$W_I$ will set the initial scores, 
votes in~$W_S$ will represent the selection of~$k$ vertices,
and finally, votes in~$W_C$ will be responsible for checking that the 
selected vertices are pairwise neighboring.
We construct~$W$ such that the following will hold
for some (even) integer~$K$, determined later:

\vspace{2pt}
\begin{tabular}{l}
\label{init_scores}
$\mathrm{score}(r,W)=0$, \\
$\mathrm{score}(a,W)=K+1$ for each~$a \in A$, \\
$\mathrm{score}(q,W)=K$ for each~$q \in  \mathcal{C} \setminus (A \cup D \cup \{r\})$, \\
$\mathrm{score}(d,W) \leq 1$ for each~$d \in D$.
\end{tabular} 
\vspace{2pt}

We define the cost function~$c$ such that each swap has cost~$1$ or~$1+\varepsilon$. 
We will define each cost to be~$1$ if not explicitly stated otherwise.
Since each cost is at least~$1$, none of the candidates ranked after 
the position~$\beta+2$ in a vote~$v$ can receive non-zero points in~$v$
without violating the budget. 
Thus, we can represent votes by listing only their first~$\beta+2$ positions.
A candidate does not \emph{appear} in some vote, if 
he is not contained in these positions.

{\bf Dummies, guards, truncation, and transporters.}
First, let us clarify the concept of dummy candidates: 
we will ensure that no dummy can receive more than one point in total, by 
letting each~$d \in D$ appear in exactly one vote.
Since we will use at most one dummy in each vote,
this can be ensured easily by using at most~$|W|$ dummies in total.
We will use the sign~$\ast$ to denote dummies in votes.

Now, we define~$\beta+2$ guards using the votes~$W_G$.
We let~$W_G$ contain votes of the form~$w_G(h)$ for each~$h \in [\beta+2]$, 
each such vote having multiplicity~$K/2$ in~$W_G$. 
We let~$w_G(h)=(g_h, g_{h+1}, g_{h+2}, \dots, g_{\beta+2}, g_1, g_2, \dots g_{h-1})$~\footnote{For convenience, here we use the vector-style representation of the 
linear orders given by the voters.}
Note that~$\mathrm{score}(g,W_G)=K$ for each~$g \in G$, 
and the total score obtained by the guards in~$W_G$ cannot decrease.
As we will make sure that $p$ cannot receive more 
than~$K$ points without exceeding the budget, 
this yields that in any possible solution, each guard must have score exactly~$K$.

Using guards, we can \emph{truncate} votes at any position~$h > 2$ by 
putting arbitrarily chosen guards at the positions~$h, h+1, \dots, \beta+2$. 
This way we ensure that only candidates on the first~$h-1$ positions can receive a point in this vote. 
We will denote truncation at position~$h$ by using a sign~$\dagger$ at that position.

Sometimes we will need votes which ensure that some candidate 
$q_1$ can ``transfer'' one point to some candidate~$q_2$ using a cost of~$c$ from the budget 
($c \in \mathbb{N}^+$, $q_1,q_2 \in \mathcal{C} \setminus (D \cup D \cup G)$).
In such cases, we construct~$c$ votes
by using exactly~$c-1$ \emph{transporter} candidates, say~$t_1, t_2, \dots, t_{c-1}$, 
none of which appears in any other vote in $W \setminus W_I$. 
The constructed votes are as follows: 
for each~$h \in [c-2]$ we add a vote~$(\ast, t_h, t_{h+1}, \dagger)$, and 
we also add the votes~$(\ast, q_1, t_1, \dagger)$ and~$(\ast, t_{c-1}, q_2, \dagger)$. 
We let the cost of any swap here be~$1$, and we denote the obtained set of votes by~$q_1 \leadsto^c q_2$.
(Note that~$q_1 \leadsto^1 q_2$ only consists of the vote~$(\ast, q_1, q_2, \dagger)$.)

Observe that the votes~$q_1 \leadsto^c q_2$ ensure that~$q_1$ can transfer one point to~$q_2$ at cost~$c$. 
Later, we will make sure~$\mathrm{score}(t,W)=K$ for each transporter~$t \in T$.
Thus, no transporter can increase its score in a solution, and~$q_1$ only loses a point in these votes if~$q_2$ gets one.

{\bf Setting initial scores.}
Using dummies and guards, we define~$W_I$ to adjust the initial scores of the relevant candidates as follows.
We put the following votes into~$W_I$:

\vspace{2pt}
\begin{tabular}{p{6cm}l}
$(p, \ast, \dagger)$ with multiplicity~$K$, & \\
$(a^{i,j}, \ast, \dagger)$ with multiplicity~$K+1-|V_j|$ & for each~$i,j \in [k]$,  \\
$(h^i_x, \ast, \dagger)$ with multiplicity~$K-|E_x^i|$ &
	for each~$i \in [k], x \in \bigcup_{i<j} V_j$, \\
$(\wt{h}^i_x, \ast, \dagger)$ with multiplicity~$K-|E_x^i|$ &
	for each~$i \in [k], x \in \bigcup_{i>j} V_j$, \\
$(m^{i,j}, \ast, \dagger)$ with multiplicity~$K-2$ & for each~$i<j$,  \\
$(q, \ast, \dagger)$ with multiplicity~$K-1$ &
	for each remaining~$q \notin D \cup G \cup \{r\}$.
\end{tabular}
\vspace{2pt}

The preferred candidate~$p$ will not appear in any other vote, implying~$\score(p,W)=K$.
 
{\bf Selecting vertices.}
The set~$W_S$ consists of the following votes:

\vspace{2pt}
\begin{tabular}{p{6cm}l}
$a^{i,j} \leadsto^1 b_x^i$ & for each~$i,j \in [k]$ and~$x \in V_j$, \\
$b_x^1 \leadsto^2 \wt{c}_x^1$ & for each~$x \in V$, \\
$w_S(i,x)=(b_x^i, c_x^{i-1}, \wt{c}_x^i, f_x^{i-1}, \dagger)$ & for each~$2 \leq i \leq k$,~$x \in V$, \\
$c^{k}_x \leadsto^2 f_x^k$ & for each~$x \in V$, \\
$\wt{c}_x^i \leadsto^1 c_x^i$ & for each~$i \in [k], x \in V$, and  \\
$f_x^i \leadsto^{2(k-i)+1} h_x^i$ & for each~$i \in [k], x \in V$.
\end{tabular}
\vspace{2pt}

Swapping candidate~$b_x^i$ with~$c_x^{i-1}$, and 
swapping candidate~$\wt{c}_x^i$ with~$f_x^{i-1}$ in~$w_S(i,x)$ for some~$2 \leq i \leq k$,~$x \in V$ 
will have cost~$1+\varepsilon$. 

{\bf Checking incidency.}
The set~$W_C$ will contain the votes

\vspace{2pt}
\begin{tabular}{p{6cm}l}
$h_x^i \leadsto^1 \wt{h}_x^i$ & for each~$i \in [k]$,~$x \in \bigcup_{i>j} V_j$, \\
$w_C(i,j,y,x) = (h_x^i, \wt{h}_y^j, \wt{m}^{i,j}, m^{i,j}, \dagger)$ 
	& for each~$i < j$,~$x \in V_j$,~$y \in V_i$,~$xy \in E$, \\
$\wt{m}^{i,j} \leadsto^1 m^{i,j}$ & for each~$i<j$, \\
$h_x^i \leadsto^3 m^{i,i}$ & for each~$i \in [k]$,~$x \in V_i$, and \\
$m^{i,j} \leadsto^1 r$ with multiplicity~$2$ & for each~$i<j$, \\
$m^{i,i} \leadsto^1 r$ & for each~$i \in [k]$.
\end{tabular}	
\vspace{2pt}

Again, swapping candidate~$h_x^i$ with~$\wt{h}_y^j$, 
and also candidate~$\wt{m}^{i,j}$ with~$m^{i,j}$ in a vote of the form~$w_C(i,j,y,x)$ 
will have cost~$1+\varepsilon$.

\begin{figure}[t]
\begin{center}
    \psfrag{a1}[lB][lB]{$a^{1,j}$}
    \psfrag{a2}[lB][lB]{$a^{2,j}$}
    \psfrag{aj}[lB][lB]{$a^{j,j}$}
    \psfrag{aj+}[lB][lB]{$a^{j+1,j}$}
    \psfrag{aj++}[lB][lB]{$a^	{j+2,j}$}
    \psfrag{ak}[lB][lB]{$a^{k,j}$}
    \psfrag{b1}[lB][lB]{$b^1_x$}
    \psfrag{b2}[lB][lB]{$b^2_x$}
    \psfrag{bj}[lB][lB]{$b^j_x$}
    \psfrag{bj++}[lB][lB]{$b^{j+2}_x$}
    \psfrag{bj+}[lB][lB]{$b^{j+1}_x$}
    \psfrag{bk}[lB][lB]{$b^k_x$}
    \psfrag{c1}[lB][lB]{$c^1_x$}
    \psfrag{cj-}[lB][lB]{$c^{j-1}_x$}
    \psfrag{cj}[lB][lB]{$c^j_x$}
    \psfrag{cj+}[lB][lB]{$c^{j+1}_x$}
    \psfrag{ck-}[lB][lB]{$c^{k-1}_x$}
    \psfrag{ck}[lB][lB]{$c^k_x$}
    \psfrag{c'1}[lB][lB]{$\wt{c}^1_x$}
    \psfrag{c'2}[lB][lB]{$\wt{c}^2_x$}
    \psfrag{c'j}[lB][lB]{$\wt{c}^j_x$}
    \psfrag{c'j+}[lB][lB]{$\wt{c}^{j+1}_x$}
    \psfrag{c'j++}[lB][lB]{$\wt{c}^{j+2}_x$}
    \psfrag{c'k}[lB][lB]{$\wt{c}^k_x$}
    \psfrag{f1}[lB][lB]{$f^1_x$}
    \psfrag{fj-}[lB][lB]{$f^{j-1}_x$}
    \psfrag{fj}[lB][lB]{$f^j_x$}
    \psfrag{fj+}[lB][lB]{$f^{j+1}_x$}
    \psfrag{fk-}[lB][lB]{$f^{k-1}_x$}
    \psfrag{fk}[lB][lB]{$f^k_x$}
    \psfrag{h1}[lB][lB]{$h^1_x$}
    \psfrag{hj-}[lB][lB]{$h^{j-1}_x$}
    \psfrag{hj}[lB][lB]{$h^j_x$}
    \psfrag{hj+}[lB][lB]{$h^{j+1}_x$}
    \psfrag{hk-}[lB][lB]{$h^{k-1}_x$}
    \psfrag{hk}[lB][lB]{$h^k_x$}
    \psfrag{h'j+}[lB][lB]{$\wt{h}^{j+1}_x$}
    \psfrag{h'k-}[lB][lB]{$\wt{h}^{k-1}_x$}
    \psfrag{h'k}[lB][lB]{$\wt{h}^k_x$}
    \psfrag{m'1}[lB][lB]{$\wt{m}^{1,j}$}
    \psfrag{m'j-}[lB][lB]{$\wt{m}^{j-1,j}$}
    \psfrag{m1}[lB][lB]{$m^{1,j}$}
    \psfrag{mj-}[lB][lB]{$m^{j-1,j}$}
    \psfrag{mj}[lB][lB]{$m^{j,j}$}
    \psfrag{mj+}[lB][lB]{$m^{j+1,j}$}
    \psfrag{mk-}[lB][lB]{$m^{k-1,j}$}
    \psfrag{mk}[lB][lB]{$m^{k,j}$}
    \psfrag{r}[lB][lB]{$r$}
	\psfrag{1}[lB][lB]{$ \scriptstyle 1$}
	\psfrag{1}[lB][lB]{$ \scriptstyle 2$}
	\psfrag{1}[lB][lB]{$ \scriptstyle 3$}
	\psfrag{1}[lB][lB]{$ \scriptstyle 4$}	
\includegraphics[scale=0.6]{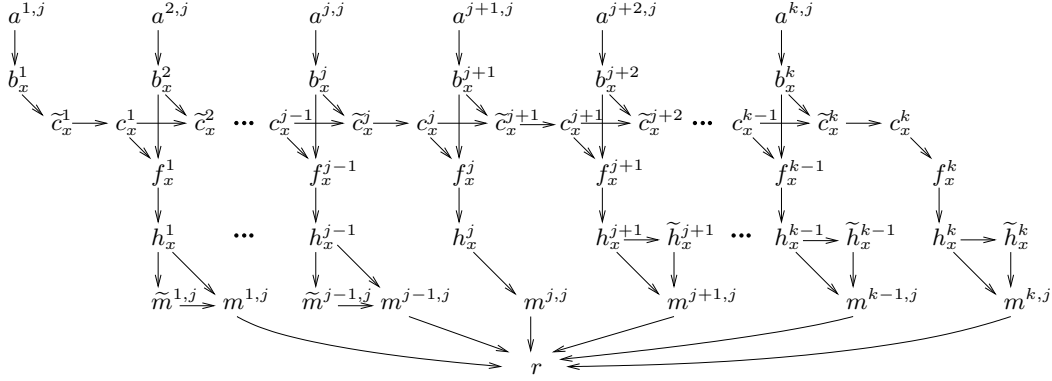}
\end{center}
\caption{Part of the instance~$I_{\mathcal{G}}$ 
in the proof of Theorem~\ref{thm_2approval}, 
assuming $x \in V_j$ in the figure. 
An arc goes from~$q_1$ to~$q_2$ if~$q_1$ can transfer a point to~$q_2$ using one or several swaps. 
}
\label{fig_instance}
\end{figure}

It remains to define~$K$ properly. To this end, we let~$K \geq 2$ be the minimum even integer not smaller than the integers 
in the set~$\{|E_x^j| \mid j \in [k], x \notin V_j\} \cup \{|V_i| \mid i \in [k]\} \cup \{k^2\}$. 
This finishes the construction.
It is straightforward to verify that the initial scores of the candidates are as claimed above.
The constructed instance is illustrated in Figure~\ref{fig_instance}.

{\bf Construction time.}
Note~$|W_G|=(\beta+2)K/2$, 
$|W_I|=O(Kk^2+Kk|V|)=O(Kk|V|)$, $|W_S|=O(k^2|V|)$, and~$|W_C|=O(k|V|+|E|)$. 
Hence, the number of votes is polynomial in the size of the input graph~$\mathcal{G}$. 
This also implies that the number of candidates is polynomial as well, 
and the whole construction takes polynomial time.
Note also that~$\beta$ is only a function of~$k$, hence this yields a parameterized reduction as well.

If for some vote~$v$, exactly one candidate~$q_1$ gains a point and 
exactly one candidate~$q_2$ loses a point as a result of the swaps in~$\Gamma$, then we say that 
$q_2$ \emph{sends} one point to~$q_1$, or equivalently, 
$q_1$ \emph{receives} one point from~$q_2$ in~$v$ according to~$\Gamma$.
Also, if~$\Gamma$ consists of swaps that transform a vote~$(a,b,c,d, \dagger)$ into a 
vote~$(c,d,a,b,\dagger)$, then we say that~$a$ sends one point to~$c$, and~$b$ sends one point to~$d$.
A point is \emph{transferred} from~$q_1$ to~$q_2$ in~$\Gamma$, if it is sent from~$q_1$ to~$q_2$
possibly through some other candidates.

Our aim is to show the following:~$\mathcal{G}$ has a~$k$-clique if and only if 
the constructed instance~$I_{\mathcal{G}}$ is a yes-instance of \textsc{Swap Bribery}. 
This will prove both (1) and (2).

{\bf Direction~$\Longrightarrow$.} 
Suppose that~$\mathcal{G}$ has a clique consisting of the vertices~$x_1, x_2, \dots, x_k$ with~$x_i \in V_i$.
We are going to define a set~$\Gamma$ of swaps transforming~$W$ into~$W^{\Gamma}$ 
with total cost~$\beta$ such that~$p$ wins in~$W^{\Gamma}$ according to 2-approval. 

First, we define the swaps applied by~$\Gamma$ in~$W_S$:
\begin{itemize}
\item Swap~$a^{i,j}$ with~$b^i_{x_j}$ for each~$i$ and~$j$ in~$a^{i,j} \leadsto^1 b^i_{x_j}$.
Cost:~$k^2$.
\item Transfer one point from~$b^1_{x_j}$ to~$\wt{c}^1_{x_j}$ for each~$j \in [k]$ in~$b_{x_j}^1 \leadsto^2 \wt{c}_{x_j}^1$.
Cost:~$2k$.
\item Apply four swaps in each 
vote~$w_S(i,x_j)=(b_{x_j}^i, c_{x_j}^{i-1}, \wt{c}_{x_j}^i, f_{x_j}^{i-1}, \dagger)$ transforming it
to $(\wt{c}_{x_j}^i, f_{x_j}^{i-1}, b_{x_j}^i, c_{x_j}^{i-1}, \dagger)$, 
sending one point from~$b_{x_j}^i$ to~$\wt{c}_{x_j}^i$ and simultaneously, also one point from~$c_{x_j}^{i-1}$ to~$f_{x_j}^{i-1}$.
Cost:~$4k(k-1)$. 
\item Swap~$\wt{c}^i_{x_j}$ with~$c^i_{x_j}$ for each~$i,j \in [k]$ in~$\wt{c}_{x_j}^i \leadsto^1 c_{x_j}^i$.
Cost:~$k^2$.
\item Transfer one point from~$c^k_{x_j}$ to~$f^k_{x_j}$ for each~$j \in [k]$ in~$c^k_{x_j} \leadsto^2 f^k_{x_j}$.
Cost:~$2k$.
\item Transfer one point from~$f^i_{x_j}$ to~$h^i_{x_j}$ for each~$i,j \in [k]$ in~$f^i_{x_j} \leadsto^{2(k-i)+1} h^i_{x_j}$.
Cost:~$k^3$.
\end{itemize}
The above swaps transfer one point 
from~$a^{i,j}$ to~$h^i_{x_j}$ via the candidates~$b^i_{x_j}$, $\wt{c}^i_{x_j}$, $c^i_{x_j}$, and~$f^i_{x_j}$ for each~$i$ and~$j$.
These swaps of~$\Gamma$, applied in the votes~$W_S$, have total cost~$k^3 +6k^2$. 

Now, we define the swaps applied by~$\Gamma$ in the votes~$W_C$. 
\begin{itemize}
\item Swap~$h_{x_j}^i$ with~$\wt{h}_{x_j}^i$  for each~$j<i$ in~$h_{x_j}^i \leadsto^1 \wt{h}_{x_j}^i$.
Cost:~$k(k-1)/2$.
\item Apply four swaps in each 
vote~$w_C(i,j,x_i,x_j) = (h_{x_j}^i, \wt{h}_{x_i}^j, \wt{m}^{i,j}, m^{i,j}, \dagger)$  transforming it
to~$(\wt{m}^{i,j}, m^{i,j}, h_{x_j}^i, \wt{h}_{x_i}^j, \dagger)$, 
sending one point from~$h_{x_j}^i$ to~$\wt{m}^{i,j}$ and, simultaneously, also one point from~$\wt{h}_{x_i}^j$ to~$m^{i,j}$.
Note that~$w_C(i,j,x_i,x_j)$ is indeed defined for each~$i$ and~$j$, 
since~$x_i$ and~$x_j$ are neighboring.
Cost:~$2k(k-1)$. 
\item Swap~$\wt{m}^{i,j}$ with~$m^{i,j}$ for each~$i<j$ in~$\wt{m}^{i,j} \leadsto^1 m^{i,j}$. 
Cost:~$k(k-1)/2$.
\item Transfer one point from~$h_{x_i}^i$ to~$m^{i,i}$ for each~$i \in [k]$ in~$h_{x_i}^i \leadsto^3 m^{i,i}$.
Cost:~$3k$.
\item Swap~$m^{i,j}$ with~$r$ in both of the votes~$m^{i,j} \leadsto^1 r$ 
for each~$i<j$. 
Cost:~$k(k-1)$.
\item Swap~$m^{i,i}$ with~$r$ for each~$i \in [k]$ in~$m^{i,i} \leadsto^1 r$. 
Cost:~$k$.
\end{itemize}

Candidate~$r$ receives~$k^2$ points after all these swaps in $\Gamma$.
Easy computations show that the above swaps have cost~$4k^2$, so the total cost of $\Gamma$ is~$\beta=k^3+10k^2$.
Clearly, 

\vspace{2pt}
\begin{tabular}{l}
$\score(p,W^{\Gamma})=K$, \\
$\score(r,W^{\Gamma})=k^2 \leq K$,  \\
$\score(a,W^{\Gamma})=K$ for each~$a \in A$, and \\
$\score(q,W^{\Gamma})=\score(q,W) \leq K$ for all the remaining candidates~$q$. 
\end{tabular}
\vspace{2pt}

This means that~$p$ is a winner in~$W^{\Gamma}$ according to 2-approval.
Hence,~$\Gamma$ is indeed a solution for~$I_{\mathcal{G}}$, proving the first direction of the reduction.

{\bf Direction~$\Longleftarrow$.} 
Suppose that~$I_{\mathcal{G}}$ is solvable, and there is a set~$\Gamma$ of swaps transforming~$W$ into~$W^{\Gamma}$ 
with total cost at most~$\beta$ such that~$p$ wins in~$W^{\Gamma}$ according to 2-approval. We also assume w.l.o.g. 
that~$\Gamma$ is a solution having minimum cost. 

As argued above,~$\mathrm{score}(p,W^{\Gamma})\leq K$ 
and~$\mathrm{score}(g,W^{\Gamma}) \geq K$ for each~$g \in G$ follow directly from the construction. 
Thus, only~$\mathrm{score}(p,W^{\Gamma}) = \mathrm{score}(g,W^{\Gamma}) = K$  for each~$g \in G$ is possible. 
Hence, for any~$i,j \in [k]$, by~$\mathrm{score}(a^{i,j},W)=K+1$ we get that~$a^{i,j}$ 
must lose at least one point during the swaps in~$\Gamma$. 
As no dummy can have more points in~$W^{\Gamma}$ than in~$W$ (by their positions), 
and each candidate in~$\mathcal{C} \setminus (A \cup D \cup \{r\})$ has~$K$ points in~$W$, 
the~$k^2$ points lost by the candidates in~$A$ can only be transferred by~$\Gamma$
to the candidate~$r$. 

By the optimality of~$\Gamma$, this means that 
$a^{i,j}$ sends a point to~$b_x^i$ in~$\Gamma$ for some unique~$x \in V_j$; 
we define~$\sigma(i,j) = x$ in this case. 
First, we show~$\sigma(1,j)=\sigma(2,j) = \dots = \sigma(k,j)$ for each~$j \in [k]$, and 
then we prove that the vertices~$\sigma(1,1), \dots, \sigma(k,k)$ 
form a~$k$-clique in~$\mathcal{G}$.

Let~$B^*$ be the set of candidates in~$B$ that receive a point from 
some candidate in~$A$ according to~$\Gamma$;
$|B^*|=k^2$ follows from the minimality of~$\Gamma$.
Observing the votes in~$W_S \cup W_C$, we can see that some~$b^i_x \in B^*$ 
can only transfer one point to~$r$ by transferring it 
to~$h^{i'}_x$ via~$f^{i'}_x$ for some~$i'$ using swaps in the votes~$W_S$, 
and then transferring the point from~$h^{i'}_x$ to~$r$ using swaps in the votes~$W_C$.
Basically, there are three ways to transfer a point from~$b^i_x$ to~$h^{i'}_x$:
\begin{itemize}
\item[(A)] $b^i_x$ sends one point to~$f^{i-1}_x$ in~$w_S(i,x)$ at a cost of~$3+2 \varepsilon$, 
and then~$f^{i-1}_x$ transfers one point to~$h^{i-1}_x$. 
This can be carried out applying exactly~$3+2(k-i+1)+1=6+2(k-i)$ swaps, 
having total costs~$6+2(k-i)+2 \varepsilon$.
\item[(B)] $b^i_x$ sends one point to~$\wt{c}^i_x$ in~$w_S(i,x)$, 
$\wt{c}^i_x$ sends one point to~$c^i_x$, 
$c^i_x$ sends one point to~$f^i_x$, and then the point gets transferred to~$h^i_x$. 
Again, the number of used swaps is exactly~$5+2(k-i)+1=6+2(k-i)$, 
and the total cost is at least~$6+2(k-i)$.
\item[(C)] $b^i_x$ sends one point to~$\wt{c}^i_x$ in~$w_S(i,x)$, and then the point is transferred 
to a candidate~$f^{i'}_x$ for some~$i' > i$ via the candidates~$c^i_x, \wt{c}^{i+1}_x, c^{i+1}_x, \dots, c^{i'}_x$. 
Again, the number of used swaps is exactly~$5+2(k-i)+1=6+2(k-i)$, 
and the total cost is at least~$6+2(k-i)$.
\end{itemize}
Summing up these costs for each~$b^i_x \in B^*$, and taking into account the cost of sending
the~$k^2$ points from the candidates of~$A$ to~$B^*$, 
we get that the swaps of~$\Gamma$
applied in the votes~$W_S$ must have total cost at least 
$k^2 + k \left( \sum_{j=1}^k 6+2(k-i) \right) =k^3 +6k^2$. 
Equality can only hold if each~$b^i_x \in B^*$ transfers one 
point to~$h^{i'}_x$ for some~$i' \geq i$, i.e. either case~B or~C happens.

Let~$H^*$ be the set of those~$k^2$ candidates in~$H$ that receive a point transferred 
from a candidate in~$B^*$, and let us consider now the swaps of~$\Gamma$ applied in the votes~$W_C$
that transfer one point from a candidate~$h^i_x \in H^*$ to~$r$. 
Let~$j$ be the index such that~$x \in V_j$.
First, note that~$h^i_x$ must transfer one point to~$m^{i,j}$ (if~$i \leq j$) or to~$m^{j,i}$ (if~$i>j$).
Moreover, independently of whether~$i<j$,~$i=j$, or~$i>j$ holds, this can only be done using exactly~$3$ swaps, 
thanks to the role of the candidates in~$\wt{H}$ and in~$\wt{M}$. 
To see this, note that only the below possibilities are possible: 

\begin{itemize}
\item If $i<j$, then $h^i_x$ sends one point in $w_C(i,j,y,x)$ for some $y \in V_i$ 
either to~$\wt{m}^{i,j}$ via two swaps, or to~$m^{i,j}$ via three swaps. 
In the former case, $\wt{m}^{i,j}$ must further transfer the point to~$m^{i,j}$, 
which is the third swap needed.
\item If $i>j$, then $h^i_x$ first sends one point to $\wt{h}^i_x$, and then 
$\wt{h}^i_x$ sends this point either to~$\wt{m}^{j,i}$ via one swap, or to~$m^{j,i}$ via two swaps
applied in the vote~$w_C(j,i,x,y)$ for some $y \in V_i$. In the former case, 
$\wt{m}^{j,i}$ transfers the point to~$m^{j,i}$ via an additional swap.
Note that in any of these cases, $\Gamma$ applies 3 swaps (maybe having cost $3+\varepsilon$ or $3+2\varepsilon$).
\item If $i=j$, then $h^i_x$ sends one point to~$m^{i,i}$ through 3 swaps.
\end{itemize}

Thus, transferring a point from~$h^i_x$ to~$r$ needs 4 swaps in total, and hence the number of swaps 
applied by~$\Gamma$ in the votes~$W_C$ is at least~$4k^2$. 
Now, by~$\beta=k^3+10k^2$ we know that equality must hold everywhere in the previous reasonings. 
Therefore, as argued above, each~$b^i_x$ must transfer a point to~$h^{i'}_x$ for some~$i'\geq i$, 
i.e., only cases~B and~C might happen from the above listed possibilities. 
Now, we are going to argue that only case~B can occur.

Let us consider the multiset $I_B$ containing $k^2$ pairs of indices, obtained by 
putting~$(i,j)$ into $I_B$ for each $b^i_x \in B^*$ with $x \in V_j$.
It is easy to see that $I_B=\{(i,j) \mid 1 \leq i,j \leq k\}$.
Similarly, we also define the multiset $I_H$ containing $k^2$ pairs of indices, obtained by 
putting~$(i,j)$ into $I_H$ for each $h^i_x \in H^*$ with $x \in V_j$.
By the previous paragraph, $I_H$ can be obtained from $I_B$ by taking some pair $(i,j)$ from $I_B$
and replacing them with corresponding pairs $(i',j)$ where $i' > i$.
Let the \emph{measure} of a multiset of pairs $I$ be $\mu(I) = \sum_{(i,j) \in I} i+j$.
Then, $\mu(I_H) \geq \mu(I_B) = k^2(k+1)$.

By the above arguments, if for some $i<j$ the pair $(i,j)$ is contained with multiplicity~$m_1$ in $I_H$, 
and $(j,i)$ is contained with multiplicity~$m_2$ in $I_H$, 
then the candidate~$m^{i,j}$ has to send $m_1+m_2$ points to $r$. 
Similarly, if $(i,i)$ is contained in $I_H$ with multiplicity $m$, then $m^{i,i}$ has to send $m$ points to $r$. 
Thus, $\mu(I_H)$ equals the value obtained by summing up $i+j$ for each $m^{i,j}$ and for
each point transferred from $m^{i,j}$ to $r$.
However, each~$m^{i,j}$ (where~$i<j$) can only send two points to~$r$, 
and each~$m^{i,i}$ can only send one point to~$r$, implying $\mu(I_H) \leq 
\sum_{i \in [k]} (i+i) + 2\sum_{1 \leq i<j \leq k} (i+j) = k^2(k+1) = \mu(I_B)$.
Hence, the measures of $I_B$ and $I_H$ must be equal, from which
$I_H=I_B$ follows. Thus, only case~B can happen.

Therefore,~$\Gamma$ must send one point from~$b^i_x$ to~$\wt{c}^i_x$ 
at a cost of~2, and apply three more swaps of cost~3 
to transfer one point from~$\wt{c}^i_x$ to~$f^i_x$. 
But in the case~$i \geq 2$, this can only be done avoiding any swap of cost~$1+\varepsilon$ 
in the vote~$w_S(i,x)$, if~$f^{i-1}_x$  simultaneously receives one point from~$c^{i-1}_x$ in~$w_S(i,x)$ as well, 
which implies~$b^{i-1}_x \in B^*$. Applying this argument iteratively, 
this shows that~$b^i_x \in B^*$ implies~$\{b^h_x \mid h<i\} \subseteq B^*$. 
Hence,~$B^*$ is the union of~$k$ sets of the form~$\{ h^1_x, h^2_x, \dots, h^k_x \}$, 
implying $\sigma(1,j)=\sigma(2,j) = \dots = \sigma(k,j)$ for each~$j \in [k]$.

Finally, consider the swaps that transfer one point from~$h^i_x \in H^*$ 
to~$m^{i,j}$ in~$W_C$ where~$x \in V_j$ and~$i<j$.
We know that if~$x \in V_j$, then this must be done by applying some swaps 
in the vote~$w_C(i,j,y,x)$ for some~$y \in V_i$ such that~$xy \in E$. 
But because of our budget, each such swap must have cost~$1$ and not~$1 + \varepsilon$, which can only happen if
$\Gamma$ transforms~$w_C(i,j,y,x) = (h_x^i, \wt{h}_y^j, \wt{m}^{i,j}, m^{i,j}, \dagger)$ 
into~$(\wt{m}^{i,j}, m^{i,j}, h_x^i, \wt{h}_y^j, \dagger)$. But this implies that~$h^j_y$ must also be in~$B^*$,
implying $y=\sigma(j,i)$. 
Therefore we obtain that~$\sigma(i,j)$ and~$\sigma(j,i)$ must be vertices connected by an edge in~$\mathcal{G}$.
This proves the existence of a $k$-clique in~$\mathcal{G}$, proving the theorem.
\qed
\end{proof}

Looking into the proof of Theorem~\ref{thm_2approval}, we can see that the results hold even 
in the following restricted case: 
\begin{itemize}
\item the costs are uniform in the sense that swapping two given candidates has the same price in any vote, and
\item the maximum number of swaps allowed in a vote is four. 
\end{itemize}

By applying minor modifications to the given reduction, 
Theorem~\ref{thm_2approval} can be generalized to hold for the following modified versions as well.
\begin{itemize}
\item If we want $p$ to be the unique winner: we only have to set $\mathrm{score}(p,W)=K+1$.
\item If we use $k$-approval for any fixed $k$ with $k \geq 3$ instead of 2-approval:
it suffices to insert $k-2$ dummies into the first $k-2$ positions of each 
vote.\footnote{
Note that the number of candidates in the constructed instance will depend on the value of $k$, 
so in particular, the result does not hold for voting rules such as veto or $(m-2)$-approval.
}
\end{itemize}

We can summarize these generalizations of Theorem~\ref{thm_2approval} in the following 
theorem, which follows directly from the discussion above.

\begin{theorem}[Generalization of Theorem~\ref{thm_2approval}]
For any constant $k \geq 2$, \textsc{Swap Bribery} for $k$-approval is W[1]-hard 
when parameterized by the value of the budget, assuming that the minimum cost of a swap is $1$; 
this holds even if the following restrictions apply:
\begin{itemize}
\item there are only two different positive costs possible for a swap, 
i.e. each swap has a cost in $\{ 1, 1+\epsilon\}$ for some $\epsilon>0$,
\item the cost of swapping two given candidates is the same in each vote, and
\item the maximum number of swaps allowed in a vote is 4. 
\end{itemize}
\end{theorem}


\section{Parameterizing with the number of candidates}
\label{sec:candidates}

In this section, we will consider 
the parameter `number of candidates'. For this case, the following definition is helpful. 

Let $S_m= \{ \pi_1, \pi_2, \dots, \pi_{m!} \}$ be the set of permutations of size $m$. 
We say that an election system is \emph{described by linear inequalities}, 
if for a given set $C= \{c_1, c_2, \dots, c_m\}$ of candidates 
it can be characterized by $f(m)$ sets $A_1, A_2, \dots A_{f(m)}$ (for some computable function~$f$) 
of linear inequalities over $m!$ variables
$x_1, x_2, \dots, x_{m!}$ in the following sense: 
if $n_i$ denotes the number of those votes in a given election~$E$ that order $C$ according to $\pi_i$, then 
the first candidate $c_1$ is a winner of the election if and only if 
for at least one index $i$, the setting~$x_j=n_j$ for each $j$ satisfies all inequalities in $A_i$.
Let us remark that Faliszewski et al. independently defined 
a very similar notion in the context of multimode control problems~\cite{FHH09ijcai}.

It is easy to see that many election systems can be described by linear inequalities: 
any system based on scoring rules, Copeland$^\alpha$ ($0 \leq \alpha \leq 1$), 
Maximin, Bucklin, Ranked pairs.
For example, $k$-approval is described by the following set $A_1$ of linear inequalities: 
\begin{eqnarray*}
A_1: \quad  \sum_{i: \rank(c_1,\pi_i) \leq k} x_i &\geq& \sum_{i: \rank(c_j,\pi_i) \leq k} x_i 	
	\quad \textrm{$\phantom{i}$for each $2 \leq j \leq m$,} 
\end{eqnarray*} 
where $\rank(c_j,\pi_i)$ denotes the position of candidate~$c_j$ in the linear order corresponding to the permutation~$\pi_i$.

To see an example where we need more than one set of linear inequalities, consider the Bucklin rule.
The \emph{Bucklin winning round} in an election is the smallest number $b$ such that there exists a
candidate that is ranked in the first $b$ positions in at least $\left \lfloor \frac{n}{2} \right \rfloor +1$ votes 
(where $n$ is the number of votes in the election). 
According to Bucklin, the winners of an election with Bucklin winning round $b$ are those candidates 
that have maximal $b$-approval score, i.e. that are ranked in the first $b$ positions by the maximum
number of votes. Note that the $b$-approval score of each winner must be
at least $\left \lfloor \frac{n}{2} \right \rfloor +1$. 
This voting system can be described by the following sets of linear inequalities $A_1, A_2, \dots, A_m$ 
where $A_b$ corresponds to the case where the Bucklin winning round is exactly $b$.

\begin{eqnarray*}
A_b: \quad \sum_{i: \rank(c_j,\pi_i) \leq b-1} x_i &\leq& \left \lfloor \frac{n}{2} \right \rfloor  
	\qquad \qquad \qquad \textrm{for each $1 \leq j \leq m$,} \\
\sum_{i: \rank(c_1,\pi_i) \leq b} x_i &\geq& \left \lfloor \frac{n}{2} \right \rfloor +1,  \\
\sum_{i: \rank(c_1,\pi_i) \leq b} x_i &\geq& \sum_{i: \rank(c_j,\pi_i) \leq b} x_i  
	\quad \textrm{$\phantom{i}$for each $2 \leq j \leq m$.} \\
\end{eqnarray*}

In the above description, the linear inequalities in the first line mean that the Bucklin winning round is at least $b$. 
The second line implies that $c_1$ has $b$-approval score at least $\lfloor \frac{n}{2} \rfloor +1$, and the 
third set of inequalities requires that no candidate has greater $b$-approval score than $c_1$. 
Clearly, $c_1$ is a winner according to Bucklin if and only if each linear inequality of $A_b$ is satisfied 
by setting~$x_j=n_j$ ($1 \leq j \leq m$) for at least one set $A_b$ among the sets $A_1, A_2, \dots, A_m$.

\begin{theorem}\label{thm_candidates}
\textsc{Swap Bribery} is FPT if the parameter is the number of candidates, for any 
election system described by linear inequalities.
\end{theorem}

\begin{proof}
Let  $C= \{c_1, c_2, \dots, c_m\}$ be the set of candidates, where~$c_1$ is the preferred one, 
and let $A_1, A_2, \dots A_{f(m)}$ be the sets of linear inequalities over variables $x_1, \dots, x_{m!}$
describing the given election system $\mathcal{E}$.
For some $\pi_i \in S_{m}$, let $v_i$ denote the vote that ranks $C$ according to $\pi_i$.
We describe the set $V$ of votes by writing $n_i$ for the multiplicity of the vote $v_i$ in $V$.

Our algorithm solves $f(m)$ integer linear programs with variables 
$T= \{t_{i,j} \mid i \neq j,$ $1 \leq i,j \leq m!\}$.
We will use $t_{i,j}$ to denote the number of votes $v_i$ that we transform into votes $v_j$;
we will require $t_{i,j} \geq 0$ for each $i \neq j$.
Let $V^T$ denote the set of votes obtained by transforming the votes in $V$ according
to the variables $t_{i,j}$ for each $i \neq j$.
Such a transformation from $V$ is feasible if 
$\sum_{j \neq i} t_{i,j} \leq n_i \textrm{ holds for each } i \in [m!]$ (inequality $\mathcal{A}$).
By 
\cite{EFS09}, we can compute the price $c_{i,j}$ 
of transforming the vote $v_i$ into $v_j$ in $O(m^3)$ time. 
Transforming $V$ into $V^T$
can be done with total cost at most $\beta$, if 
$\sum_{i,j \in [m!]} t_{i,j} c_{i,j} \leq \beta$
(inequality $\mathcal{B}$).

We can express the multiplicity $x'_i$ of the vote $v_i$ 
in $V^T$ as $x'_i= n_i +  \sum_{j \neq i} t_{j,i} - \sum_{i \neq j} t_{i,j}$.
For some $i \in [f(m)]$, let $A'_i$ denote the set of linear inequalities over the variables in $T$
that are obtained from the linear inequalities in $A_i$ by substituting $x_i$ with the above given 
expression for $x'_i$.
Using the description of $\mathcal{E}$ with the given linear inequalities,
we know that the preferred candidate $c_1$ wins in the $\mathcal{E}$-election~$(C, V^T)$ for some values of the variables $t_{i,j}$
if and only if these values satisfy the inequalities of $A'_i$ for at least one $i \in [f(m)]$.
Thus, our algorithm solves \textsc{Swap Bribery} by finding a non-negative assignment for the variables in~$T$
that satisfies both the inequalities $\mathcal{A}$, $\mathcal{B}$,
and all inequalities in $A'_i$ for some $i$.

Solving such a system of linear inequalities can be done in linear FPT time, 
if the parameter is the number of variables~\cite{lenstra-ip}. By $|T|=(m!-1)m!$ the theorem follows.  
%
\qed
\end{proof}

Similarly, we can also show fixed-parameter tractability for other problems if the parameter is the number of candidates, 
e.g. for {\sc Possible Winner} (this was already obtained by Betzler et al. for several voting systems, \cite{BHN09}), 
{\sc Manipulation} (both for weighted and unweighted voters), several variants of {\sc Control} 
(this result was obtained for Llull and Copeland voting by Faliszewski et al., \cite{FHHR09}), 
or {\sc Lobbying}~\cite{CFRS07} (here, the parameter would be the number of issues in the election). 
Since our topic is {\sc Swap Bribery}, we omit the details.

\section{Parameterizing with the number of votes}
\label{sec:votes}

In this section, we consider the case where the number of votes is a parameter. 
First, note that if $k$ is unbounded and is part of the input, then {\sc Swap Bribery} 
is NP-complete even for a single vote~\cite[Theorem 4.5]{EFS09}.
Hence, we consider parameterizations of {\sc Swap Bribery} where not only the number~$n$ of votes, 
but also either $k$ or the budget~$\beta$ is regarded as a parameter. 


We first recall that there is a simple brute force algorithm given in~\cite{EFS09}
for \textsc{Swap Bribery} that runs in $m^{O(kn)}$ time.
Looking at this running time, one can wonder whether it is possible 
to get $k$ or $n$ out of the exponent of $m$.
Theorem~\ref{thm_color}, which makes use of the technique of color-coding~\cite{alon-yuster-zwick-colorcoding},
answers this question in the affirmative for the case of $n$, 
by providing an algorithm for \textsc{Swap Bribery} for $k$-approval which is
fixed-parameter tractable with parameter $n$, supposing that $k$ is some fixed constant. 
Note that this result is best possible in the sense that the problem 
without parameterization remains NP-hard even if $k=2$.

By contrast, we will see in Theorem~\ref{thm_k_hard} that we cannot expect a similar result 
for the case where $n$ is constant but $k$ is a parameter. 

%
%

\begin{theorem}
\label{thm_color}
\textsc{Swap Bribery} for $k$-approval can be solved with a randomized algorithm in $2^{2(\log n + \log k)nk}O(m^{k+1})$ 
expected time. The derandomized version of the algorithm has running time $2^{O((\log n + \log k)nk)}O(m^{k+1} \log m)$.
\end{theorem}

\begin{proof}
We are going to apply the idea of color-coding~\cite{alon-yuster-zwick-colorcoding} 
widely used to design parameterized algorithms.
Let $I$ be the given instance of \textsc{Swap Bribery} with 
$V=\{v_1, \dots, v_n\}$ and $C=\{p, c_1, \dots, c_{m-1}\}$ denoting the set of votes and the set of candidates, 
respectively, where $p$ is our preferred candidate. 

To begin, let us introduce some definitions that capture the structure of a solution.
Let us call any $k$-size subset of the set $[nk]= \{1, \dots, nk \}$ a \emph{vote pattern}, 
and let us call an $n$-tuple of vote patterns an \emph{election pattern}.
We say that an election pattern $\mathcal{P}=(P_1, \dots, P_n)$ is \emph{successful}, 
if the element 1 appears at least as many times in $\mathcal{P}$ as any other element, i.e. if 
$|\{ i \mid 1 \in P_i \}| \geq |\{ i \mid c \in P_i \}|$ for any $c \in [nk]$.
Intuitively, we can think of an election pattern as the encoding of 
the family of those candidate sets that are moved into the first $k$ positions of some vote by a solution; 
hence, we use the $nk$ integers in correspondence to the relevant candidates obtaining at 
least one point in the bribed election. To explain the exact connection between 
election patterns and briberies, we need some additional concepts.

For a set $\Gamma$ of swaps for $I$, let $C^{rel}(\Gamma)$ denote the set of candidates 
whose score in $V^{\Gamma}$ is at least $1$. Clearly, $|C^{rel}(\Gamma)| \leq kn$ always holds, and 
if $\Gamma$ is a solution for $I$, then we also have $p \in C^{rel}(\Gamma)$.
We say that $\Gamma$ is \emph{compatible} with an election pattern $(P_1, \dots, P_n)$, 
if there is an injective function $\sigma$ mapping the elements of $C^{rel}(\Gamma)$ to 
different integers in $[nk]$ with $\sigma(p)=1$ 
such that for each vote $v_i \in V$, the set of integers assigned by $\sigma$ to the first $k$ candidates 
in the vote $v_i^{\Gamma}$ is exactly the set $P_i$.
We say that the mapping $\sigma$ is a \emph{witness} for this compatibility.

The importance of these definitions relies on the following two observations. 
On the one hand, if $\Gamma$ is a solution, then any election pattern compatible with $\Gamma$ is successful.
On the other hand, if a bribery is compatible with a successful election pattern 
and its cost does not exceed the given budget, then it yields a solution for $I$.
Therefore, our algorithm does the following: it enumerates every possible successful election pattern, and
for each such pattern it looks for the cheapest bribery compatible with it.
Note that there are at most $\binom{nk}{k}^n < (nk)^{nk}$ possible election patterns to check.

Given a successful election pattern $\mathcal{P}=(P_1, \dots, P_n)$, let $A= \bigcup_{i \in [n]} P_i \setminus \{1\}$.
We describe an algorithm $\mathcal{A}$ that, 
assuming that there exists a solution compatible with $\mathcal{P}$,
finds a solution in $(nk)^{nk}O(m^{k+1})$ randomized time.
So let us suppose from now on that $I$ admits a solution $\Gamma$ 
that is compatible with $\mathcal{P}$, and let $\sigma$ be a witness for this.
(Note that we do not know $\sigma$.)
Our algorithm applies color-coding as follows: it colors each candidate in $C \setminus \{p\}$ 
with the colors of $A$ randomly using  a uniform and independent distribution. 
Let $\alpha(c)$ denote the color of a candidate $c$; we set $\alpha(p)=1$.
We say that the coloring $\alpha$ is \emph{nice} if $\alpha(c)=\sigma(c)$ for each $c \in C^{rel}(\Gamma)$.
Clearly, a random coloring $\alpha$ is nice with probability at least $|A|^{-|A|} \geq (nk-1)^{-(nk-1)}$.

Assuming that we have a nice coloring $\alpha$, we can find a solution for $I$ as follows. 
For each $v_i \in V$, we take every possible subset $C' \subseteq C$ 
of size $k$ whose colors correspond to the vote pattern $P_i$, i.e. such that $\bigcup_{c \in C'} \alpha(c)=P_i$ holds. 
For each such $C'$, we compute the minimum cost of a bribery that moves 
the candidates of $C'$ to the first $k$ positions in $v_i$. 
This can be done in $O(mk)$ time for some $C'$, by simply swapping each candidate 
$c' \in C'$ with exactly those candidates in $C \setminus C'$ that precede $c'$ in the vote $v_i$. 
Now, for each $v_i \in V$ we take the cheapest one among all these briberies 
over all possible sets $C'$ colored by the colors of $P_i$; 
let $B_i$ be the resulting bribery for $v_i$. 
We claim that the union of the swaps in $B_1, \dots, B_n$ is a bribery $B$ that yields a solution.
Observe that $B$ can be computed in $n \binom{m}{k}O(mk)$ time. 

To prove our claim, first note that by our assumptions that $\Gamma$ is a solution compatible with $\mathcal{P}$ 
and $\alpha$ is nice, we get that the bribery $B$ cannot have cost greater than the cost of~$\Gamma$, 
as the algorithm must have considered the restriction of $\Gamma$ on $v_i$ 
when choosing $B_i$ for some $i$. Thus, $B$ does not exceed the budget. 
It remains to show that $p$ is a winner in $V^B$. 
First, observe that if $B$ is compatible with $\mathcal{P}$, then this follows from the fact that $\mathcal{P}$ is 
a successful election pattern. Unfortunately, it might happen that $B$ is not compatible with $\mathcal{P}$;
the reason for this is that different candidates in $C^{rel}(B)$ might have the same color. 
However, this will not cause any problems, since the score of any candidate $c \in C^{rel}(B)$ in $V^B$
is upper bounded by the number of occurrences of the element $\alpha(c)$ in the election pattern $\mathcal{P}$.
Since $\mathcal{P}$ is successful, this latter cannot be greater than the score of $p$ in $V^B$.
In other words, $p$ is a winner in $V^{B}$ because for each $c \in C \setminus \{p\}$ we have
$$\score(c,V^B) \leq |\{ i: \alpha(c) \in P_i\}| \leq  |\{ i: 1 \in P_i\}| = \score(p,V^B)
.$$

With this method, if the coloring is nice, then algorithm $\mathcal{A}$ finds a solution. 
By the above arguments, the randomized version of algorithm $\mathcal{A}$ 
finds a solution in $(nk)^{nk}O(m^{k+1})$ expected time, provided that there 
exists a solution compatible with $\mathcal{P}$. 
Therefore, checking every possible successful election pattern 
takes $(nk)^{2nk} O(m^{k+1})=2^{2(\log n + \log k)nk}O(m^{k+1})$ randomized time.

To derandomize the algorithm, we can apply standard techniques using $nk$-perfect
hash functions~\cite{alon-yuster-zwick-colorcoding}
instead of randomly coloring the candidates of $C \setminus \{p\}$.	
This yields a deterministic algorithm with $2^{O((\log n + \log k)nk)}O(m^{k+1} \log m)$ running time.
\qed
\end{proof}

Next, we complement Theorem~\ref{thm_color} by proving that there is no hope for getting $k$ out
of the exponent of $m$ in any algorithm solving \textsc{Swap Bribery} for $k$-approval, 
as this problem remains W[1]-hard with parameter $k$ even in the case $n=1$, 
i.e. if there is only one vote in the instance.

\begin{theorem}
\label{thm_k_hard}
\textsc{Swap Bribery} for $k$-approval with only one vote is W[1]-hard when parameterized by $k$.
\end{theorem}

\begin{proof}
We provide a parameterized reduction from the W[1]-hard \textsc{Clique} problem, 
parameterized by the size of the desired clique.
Let $G=(V,E)$ be the input graph given with $V=\{v_1, \dots, v_N\}$, and let $k$ be the parameter given. 
We are going to construct an instance~$I$ of \textsc{Swap Bribery} for $(k+1)$-approval 
consisting of an election with a single vote 
$$w: d_1 \succ \dots \succ d_{k+1} \succ c_1 \succ \dots \succ c_N \succ p,$$
a cost function $c$, and a budget $\beta=(N-k)N^2+kN-\binom{k}{2}$. 
We let $C=\{ d_1, \dots, d_{k+1}, c_1, \dots, c_N, p\}$ denote the set of candidates, 
and we let $C_v=\{ c_1, \dots, c_N \}$. 
Our preferred candidate is $p$.

The values of the cost function, shown also in Table~\ref{table_cost}, are as follows. 
(For simplicity, the cost of swapping two candidates $x_1$ and $x_2$ in the vote $w$ is denoted 
by $c(x_1,x_2)$ instead of $c(x_1,x_2,w)$, as there is only one voter.)
We define the cost of swapping $c_i$ with $c_j$ for some $i,j \in [N]$ to be $1$ if $v_i v_j$ is an edge in $G$, 
and $0$ otherwise. We set the cost of swapping $d_1$ with $c_i \in C_v$ to be $N-|\{j<i \mid v_j v_i \in E\}|$. 
Furthermore, we let the cost of swapping $p$ with any candidate $c_i \in C_v$ to be $N^2$.
All remaining swaps have zero cost. The construction takes time polynomial in $|V|+k$.

\begin{table}[t]
\begin{center}
\begin{tabular}{l@{\hspace{10pt}}l}
\noalign{\hrule}
\\[-8pt]
$c(c_i,c_j) = 1$ 			& for each $i,j \in [N]$ where $v_i v_j \in E$ 	\\
$c(c_i,c_j) = 0$ 			& for each $i,j \in [N]$ where $v_i v_j \notin E$ 	\\
$c(d_1,c_i) = N-|\{j<i \mid v_j v_i \in E\}|$ 			& for each $i \in [N]$ \\
$c(c_i,p) = N^2$ 			& for each $i \in [N]$ \\
$c(d_1,q)=0$ 				& for any $q \in C \setminus C_v$ \\
$c(d_j,q)=0$ 				& for any $2 \leq j \leq k+1$ and $q \in C$ 
\\[1mm]
\noalign{\hrule}
\end{tabular}
\end{center}
\caption{The values of the cost function $c$ in the proof of Theorem~\ref{thm_k_hard}.}
\label{table_cost}
\end{table}

We claim that the constructed instance $I$ is equivalent with the input of \textsc{Clique} in the sense
that $I$ has a solution if and only if $G$ has a clique of size $k$.

First, note that $\beta<(N-k+1)N^2$, which implies that any solution $\Gamma$ can swap $p$ with at most 
$N-k$ candidates from $C_v$. Therefore, $p$ can only obtain a point in $w^{\Gamma}$ 
if $\rank(p,w^{\Gamma})=k+1$ and there exist $k$ candidates $c_{i_1}, \dots, c_{i_k}$ 
that precede $p$. The cost of a minimum bribery achieving this 
is $$(N-k)N^2 + kN - \sum_{1 \leq h<j \leq k} c(c_{i_h},c_{i_j}).$$ The first term in this sum 
is the cost of swapping $p$ with the candidates $C \setminus \{c_{i_1}, \dots, c_{i_k},p\}$,
and the remaining terms correspond to swapping the candidates $c_{i_1}, \dots, c_{i_k}$ 
with every other candidate before them. 
Note that by the definition of the cost function, swapping $c_{i_j}$ with \emph{every} candidate preceding it in $w$
has cost $N$, but we do not swap $c_{i_j}$ with the candidates $c_{i_1}, \dots, c_{i_{j-1}}$.
Thus, it is clear that the cost of such a bribery is at most $\beta$ if and only if
$c(c_{i_h},c_{i_j})=1$ for every $1 \leq h<j \leq k$. This holds if and only if the set 
$\{v_{i_1}, \dots, v_{i_k}\}$ is a clique in $G$, showing that a solution exists if and only if 
there is a clique of size $k$ in $G$. This implies the correctness of the reduction.
\qed
\end{proof}

Finally, we present a kernelization algorithm for the \textsc{Swap Bribery} problem for $k$-approval,
where we consider both the budget $\beta$ and the number $n$ of votes as parameters.

\begin{theorem}
\label{thm_kernel}
If the minimum cost is $1$, then \textsc{Swap Bribery} for $k$-approval (where $k$ is part of the input)
with combined parameter $(n,\beta)$ admits a kernel with $O(n^2 \beta)$ votes and $O(n^2 \beta^2)$ candidates. 
Here, $n$ is the number of votes and~$\beta$ is the budget.
\end{theorem}

\begin{proof}
Let $V$, $C$, $p \in C$, and $\beta$ denote the set of votes, the set of candidates, 
the preferred candidate, and the budget given, respectively; we write $|V|=n$.
The idea of the kernelization algorithm is that not all candidates are interesting
for the problem: only candidates that can be moved within the budget $\beta$ from a zero-position
to a one-position or vice versa are relevant.

Let $\Gamma$ be a set of swaps with total cost at most $\beta$.
Clearly, as the minimum possible cost of a swap is $1$, 
we know that there are only $2\beta$ candidates $c$ in a vote $v \in V$  
for which $\score(c,v) \neq \score(c,v^{\Gamma})$ is possible, 
namely, such a $c$ has to fulfill $k-\beta+1 \leq \rank(c,v) \leq k+\beta$.
Thus, there are at most $2\beta n$ candidates for which $\score(c,V) \neq \score(c,V^{\Gamma})$ is possible;
let us denote the set of these candidates by $\widetilde{C}$. 
Let $c^*$ be a candidate in~$C \setminus \widetilde{C}$ whose score is the maximum among the 
candidates in~$C \setminus \widetilde{C}$. 

Note that a candidate $c \in C \setminus (\widetilde{C} \cup \{c^*,p\})$
has no effect on the answer to the problem instance. Indeed,   
if $\score(p,V^{\Gamma}) \geq \score(c^*,V^{\Gamma})$, then the score of $c$ is not relevant,
and conversely, if $\score(p,V^{\Gamma}) < \score(c^*,V^{\Gamma})$ then $p$ loses anyway.
Therefore, we can disregard each candidate in $C \setminus \widetilde{C}$ except for $c^*$ and $p$.

The kernelization algorithm constructs an equivalent instance $K$ as follows.
In $K$, neither the budget, nor the preferred candidate 
will be changed. However, we will change the value of $k$ to be $\beta+1$, so the kernel 
instance $K$ will contain a $(\beta+1)$-approval 
election\footnote{We use $\beta+1$ instead of $\beta$ to avoid complications with the case $\beta=0$.}. 
We define the set $V_K$ of votes and the set $C_K$ of candidates in $K$ as follows.

\begin{figure}[t]
\begin{center}
    \psfrag{beta}[][]{$\underbrace{\phantom{555555555}}_{\beta}$}
    \psfrag{k}[][]{$\overbrace{\phantom{II555555555555555555555555}}^{k}$}
    \psfrag{s}[][]{$\ast$}
    \psfrag{dots}[][]{$\cdots$}
\includegraphics{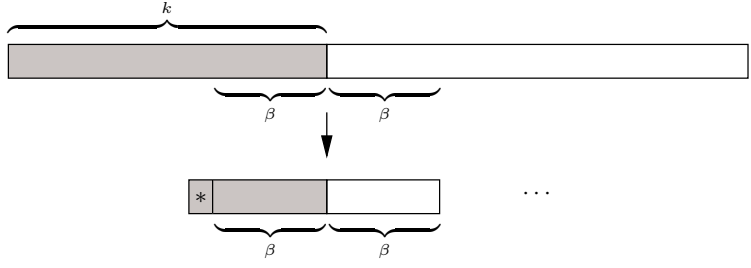}
\end{center}
\caption{Illustration of the truncation of the votes by the algorithm of Theorem~\ref{thm_kernel}.
The grey and white areas show the one-positions and the zero-positions in the depicted vote, respectively.
The symbol $\ast$ stands for a dummy candidate.}
\label{fig_kernel}
\end{figure}

First, the algorithm ``truncates'' each vote $v$, by deleting all its positions 
(together with the candidates in these positions) except for the $2\beta$ positions 
between $k-\beta+1$ and $k+\beta$. 
Then again, we shall make use of dummy candidates (see the proof of Theorem~\ref{thm_2approval});
we will ensure $\score(d, V^{\Gamma}) \leq 1$ for each such dummy $d$.
Swapping a dummy with any other candidate will have cost 1 in $K$.
Now, for each obtained truncated vote, 
the algorithm inserts a dummy candidate in the first position, so that the obtained votes have length~$2\beta+1$.
In this step, the algorithm also determines the set $\widetilde{C}$
and the candidate $c^*$. This can be done in linear time.
We denote the votes\footnote{In fact, these cropped votes are not real votes 
yet in the sense that they do not contain each candidate.} obtained in this step by $V_r$.
We do not change the costs of swapping candidates of $\widetilde{C} \cup \{c^*,p\}$ in some vote $v \in V_r$.
For an illustration of this step, see Figure~\ref{fig_kernel}.

Next, to ensure that $K$ is equivalent to the original
instance, the algorithm constructs a set~$V_d$ of votes such that
$\score(c, V_r \cup V_d)=\score(c,V)$ holds for each candidate $c$ in $\widetilde{C} \cup \{ p, c^*\}$.
This can be done by constructing $\score(c,V) - \score(c,V_r)$ newly added votes
where $c$ is on the first position, and all the next $2\beta$ positions are taken by dummies. 
This way we ensure $\score(c,V_d)=\score(c,V_d^{\Gamma})$
for any set $\Gamma$ of swaps with total cost at most $\beta$.

If $D$ is the set of dummy candidates created so far, then let $C_K=\widetilde{C} \cup \{ p, c^*\} \cup D$.
To finish the construction of the votes, 
it suffices to add for each vote $v \in V_r \cup V_d$ the candidates not yet contained in $v$,
by appending them at the end (starting from the $(2\beta+1)$-th position) in an arbitrary order.
The obtained votes will be the votes $V_K$ of the kernel.

The presented construction needs polynomial time. 
Using the above mentioned arguments, it is straightforward to verify that the constructed 
kernel instance is indeed equivalent to the original one. 
Thus, it remains to bound the size of $K$.

Clearly, $|\widetilde{C} \cup \{ p, c^*\}|\leq 2n \beta +2$. The number of dummies introduced in the first phase 
is exactly $|V_r|=n$. As the score of any candidate in $V$ is at most $n$, 
the number of votes created in the second phase is 
at most $(2n \beta +2)n$, which implies that 
the number of dummies created in this phase is at most $(2n \beta +2)n \cdot 2\beta$.
This shows $|C_K| \leq$  $n+(2n \beta +2)(2n\beta+1)=O(n^2 \beta^2)$, 
and also $|V_K| \leq (2n \beta +3)n=O(n^2 \beta)$. 
\qed
\end{proof}

We remark that if each cost is at least $1$, then a kernel with $(k+ \beta)n$ candidates 
and~$n$ votes is easy to obtain, 
by simply deleting every candidate from the instance whose rank is greater than $k+\beta$ in all of the votes.
This simple method might be favorable to the above result in cases where~$k$ is small.


\section{Conclusion}\label{sec:conclusion}


We have taken the first step towards parameterized and multivariate investigations of 
{\sc Swap Bribery} under certain voting systems, focusing on $k$-approval. 
We discussed how the complexity of this problem depends on the cost function. 
In response to an initiation of Elkind et al.~\cite{EFS09} to identify natural cases of {\sc Swap Bribery}
that are computationally tractable, we showed that the case where all swaps have equal costs is polynomial-time solvable.
By contrast, as soon as we have two different costs, 
the problem becomes NP-complete for $k$-approval for any fixed $k \geq 2$,  
and even W[1]-hard if the parameter is the budget~$\beta$. 
 
We provided a rather general result showing that {\sc Swap Bribery} is FPT 
for a very large class of voting systems if the parameter is the number of candidates. 
This revaluates previous NP-hardness results: \textsc{Swap Bribery} could be computed efficiently 
if the number of candidates is small, which is a common setting e.g. in presidential elections.
The technique used can be applied to different problems from voting theory, 
leading to fixed-parameter tractability with respect to the number of candidates in various settings.

We also shed some light on the complexity of \textsc{Swap Bribery} for $k$-approval 
when considering combined parameters.
We hope that our results will help to understanding the intricate issues of the interplay 
between the parameters 'number of votes', the budget, or the value of $k$.
On one hand, we strengthened the known NP-completeness result for a single vote by showing W[1]-hardness 
with respect to~$k$ in the case when $k$ is part of the input.
On the other hand, we proposed an FPT algorithm for the case where the 
parameter is the number of votes, but~$k$ is a fixed constant.
In addition, we presented a polynomial kernel for the problem where the parameters are 
the number of votes and the budget.

There are plenty of possibilities to carry on our initiations. 
First, there are more parameterizations to be studied in the spirit of Niedermeier~\cite{Nie10}.
Examining the possibilities for kernelizations with respect to different parameters, 
as for instance was done by Betzler in~\cite{Bet10}, is an interesting approach. 

Second, our FPT result for the case where the parameter is the number of votes
relies on an integer linear program formulation, and uses a result by Lenstra. 
Since this approach does not provide running times that are suitable in practice, 
it would be interesting to give combinatorial algorithms that compute an optimal swap bribery. 
This might be particularly relevant for a scenario described by Elkind et al.~\cite{EFS09}, where bribery 
is not necessarily considered as an undesirable thing, like in the case of campaigning.

Also, we have focused our attention to $k$-approval, but the same questions could be studied for 
other voting systems, or for the special case of {\sc Shift Bribery} which was shown to be 
NP-complete for several voting systems~\cite{EFS09}, or other variants of the bribery problem 
as mentioned in the introduction.
For instance, we have only looked at {\it constructive} swap bribery, but the case of 
{\it destructive} swap bribery (when our aim is to achieve that a disliked candidate does {\it not} win) 
is worth further investigation as well.

\vspace{4pt}
\noindent{\bf Acknowledgments.} We thank Rolf Niedermeier for an inspiring initial discussion.



\bibliographystyle{abbrv}
\bibliography{bribery}


\end{document}